\documentclass[pra,superscriptaddress,floatfix, 11pt, letterpaper]{revtex4}
\usepackage{graphicx}
\usepackage{subfigure}
\usepackage{amssymb}
\usepackage{verbatim}
\usepackage{amsmath}
\usepackage{amsfonts}
\usepackage{amscd}
\usepackage{amssymb}
\usepackage{setspace}
\usepackage{hyperref}

\newtheorem{theorem}{Theorem}
\newtheorem{proposition}[theorem]{Proposition}

\newtheorem{lemma}[theorem]{Lemma}

\newtheorem{corollary}[theorem]{Corollary}

\newtheorem{remark}[theorem]{Remark}
\newtheorem{example}[theorem]{Example}

\newenvironment{proof}{{\noindent{\bf Proof:}}}{$\hfill\Box$}

\newcommand{\ket}[1]{|#1\rangle}
\newcommand{\bra}[1]{\langle#1|}

\def\Pr{\text{Pr}}

\def\det{\text{det}}
\newcommand{\tr}{\text{\rm{tr}}}

\begin{document}

\singlespacing

\title{A New Quantum Data Processing Inequality}

\author{Salman Beigi}
\affiliation{School of Mathematics, Institute for Research in Fundamental Sciences (IPM), P.O.\ Box 19395-5746, Tehran, Iran}

\begin{abstract}
Quantum data processing inequality bounds the set of bipartite states that can be generated by two far apart parties under local operations; having access to a bipartite state as a resource, two parties cannot locally transform it to another bipartite state with a mutual information greater than that of the resource state. Nevertheless, due to the additivity of quantum mutual information under tensor product, the data processing inequality gives no bound when the parties are provided with an arbitrary number of copies of the resource state. In this paper we introduce a measure of correlation on bipartite quantum states, called maximal correlation, that is not additive and gives the same value  when computed for multiple copies. Then, by proving a data processing inequality for this measure, we find a bound on the set of states that can be generated under local operations even when an arbitrary number of copies of the resource state is available.      
\end{abstract}

\date{March 8, 2023}

\maketitle




\section{Introduction}
Let $\rho_{AB}$ be a bipartite quantum state on registers $A$ and $B$, and assume that an arbitrary number of copies of $\rho_{AB}$ are shared between two parties Alice and Bob. The goal of Alice and Bob is to generate some bipartite state $\sigma_{EF}$ under local operations but without communication. That is for some $n$, they want to apply local super-operators $\Phi_{A^n\rightarrow E}$ and $\Psi_{B^{n}\rightarrow F}$ such that 
$$\Phi\otimes \Psi (\rho_{AB}^{\otimes n})= \sigma_{EF}.$$
Typical examples of this problem are entanglement distillation and common randomness distillation under local operations, in which case $\sigma_{EF}$ is an ebit or one bit of shared randomness.

To answer this question one cannot look for such local operators by brute-force search since we assume $n$, the number of copies of the resource state $\rho_{AB}$ is arbitrarily large. On the other hand to obtain some necessary conditions on the existence of $\Phi$ and $\Psi$ one may 
compare the strength of correlations of $\rho_{AB}^{\otimes n}$ and $\sigma_{EF}$. If $\sigma_{EF}$ is more correlated than $\rho_{AB}^{\otimes n}$, then such operations do not exists since local transformations do not generate correlation. Nevertheless again since $n$ can be arbitrarily large, the standard measures of  correlation provide us with no bound. For instance the data processing inequality of mutual information states that if $\rho_{AB}^{\otimes n}$ can be locally transformed to $\sigma_{EF}$ then
\begin{align}\label{eq:add}
nI(A; B)_{\rho} = I(A^n; B^n)_{\rho^{\otimes n}}  \geq I(E; F)_{\sigma},
\end{align}
where $I(\cdot , \cdot)$ denotes the quantum mutual information. This inequality is loose for sufficiently large $n$ and gives us no bound if $I(A; B)_{\rho}\neq 0$.

In the classical case (where $\rho_{AB}$ and $\sigma_{EF}$ are bipartite random variables) there is a measure of correlation called \emph{Hirschfeld-Gebelein-R\'enyi maximal correlation} or simply the maximal correlation~\cite{Hirschfeld, Gebelein, Renyi1, Renyi2}. 
Maximal correlation has two main properties that are useful for the problem of local state transformation described above. First, it is \emph{not} additive on independent copies of a bipartite distribution, and indeed gives the same number when computed on independent copies. Second, it satisfies a data processing inequality. Using these two properties maximal correlation gives a non-trivial bound on our problem in the classical case.

The main contribution of this paper is to generalize maximal correlation to the quantum case.

\section{A new measure of correlation}\label{sec:new}
Let $\mathcal{H}_A$ be the Hilbert space corresponding to a quantum register $A$, which in this paper  is assumed to be finite dimensional. Denote the space of linear operators acting on $\mathcal{H}_A$ by $\mathbf{L}(\mathcal{H}_A)$. Similarly define $\mathbf{L}(\mathcal{H}_B)$ and equip these two spaces with the Hilbert-Schmidt inner product, i.e., $\langle M, N\rangle = \tr(M^{\dagger}N)$. This inner product induces a norm on the space of linear operators which we denote by $\|\cdot\|_2$.

For a bipartite quantum state $\rho_{AB}$ we define its \emph{maximal correlation} by 
\begin{align}
\mu(\rho_{AB})=\max\,\, & \vert \tr(\rho_{AB} X_A\otimes Y^{\dagger}_{B})\vert\nonumber\\
& \tr(\rho_AX_A) = \tr(\rho_B Y_B) =0,\label{eq:con1}\\
& \tr(\rho_A X_AX_A^{\dagger}) = \tr(\rho_B Y_BY_B^{\dagger}) =1.\label{eq:con2}
\end{align}
Here $\rho_A$ and $\rho_B$ are the reduced density matrices on subsystems $A$ and $B$ respectively, and $X_A\in \mathbf{L}(\mathcal{H}_A)$ and $Y_B\in \mathbf{L}(\mathcal{H}_B)$. Indeed, $\mu(\rho_{AB})$ is the maximum of the (absolute value of the) expectation of the tensor product of two local operators that have zero expectation and variance $1$. 

Maximal correlation in the classical case, where $\rho_{AB}$ is a joint distribution, is first introduced by
Hirschfeld~\cite{Hirschfeld} and Gebelein~\cite{Gebelein} and then studied by R\'enyi~\cite{Renyi1, Renyi2}. Witsenhausen in~\cite{Witsenhausen} proved that maximal correlation of several independent copies of a joint distribution equals to that of a single copy. This parameter has recently been revisited by several authors; see e.g.~\cite{KangUlukus, Kumar, Polyanskiy, KamathAnantharam}.

To study properties of $\mu(\rho_{AB})$ let us define
$$\widetilde \rho_{AB} = (I_A\otimes \rho_B^{-1/2}) \rho_{AB} (\rho_A^{-1/2}\otimes I_B),$$
where inverses of $\rho_A$ and $\rho_B$ are defined on their supports. Note that $\widetilde \rho_{AB}$ is not even hermitian, so is not a density matrix.

\begin{theorem}\label{thm:schmidt}
$\mu(\rho_{AB})$ is equal to the \emph{second} Schmidt coefficient of $\widetilde \rho_{AB}$ as a vector in the bipartite Hilbert space $\mathbf{L}(\mathcal{H}_A)\otimes \mathbf{L}(\mathcal{H}_B)$.
\end{theorem}

\begin{proof}  
Let $R_A=\rho_A^{1/2}X_A$ and $S_B = Y_B^{\dagger} \rho_B^{1/2}$. With this change of variables $\mu(\rho_{AB})$ is equivalently equal to 
\begin{align}
\mu(\rho_{AB}) = \max\,\,  &  |\tr(\widetilde \rho_{AB} R_A\otimes S_B) |\nonumber\\
&  \langle \rho_A^{1/2} ,  R_A\rangle = \langle \rho_B^{1/2}, S_B\rangle =0,\label{eq:1}\\
& \|R_A\|_2=\|S_B\|_2=1.\label{eq:2}
\end{align}

Let  
$$\widetilde \rho_{AB} = \sum_{i} \lambda_i M_{i}\otimes N_{i},$$
be the Schmidt decomposition of $\widetilde \rho_{AB}\in \mathbf{L}(\mathcal{H}_A)\otimes \mathbf{L}(\mathcal{H}_B)$ where $\lambda_1\geq \lambda_2\geq \cdots \geq 0$ are the Schmidt coefficients and $\{ M_i\}$ and $\{ N_i\}$ are orthonormal bases for $\mathbf{L}(\mathcal{H}_A)$ and $\mathbf{L}(\mathcal{H}_B)$, respectively. 
Note that
$$\lambda_1 = \max_{\|V\|_2=\|W\|_2=1}  \tr\left( \widetilde \rho_{AB} V_A\otimes W_B \right),  $$
and using Cauchy-Schwarz inequality we have
\begin{align*}
\lambda_1 & = \tr\left(   M_{1}^{\dagger}\otimes N_1^{\dagger} \widetilde \rho_{AB}  \right)\\
& = \tr\left[       \left(\rho_A^{-1/2}M_1^{\dagger}\right)\otimes \left(N_1^{\dagger}\rho_B^{-1/2}\right) \rho_{AB}   \right]\\
& = \tr\left[   \left(      \rho_{AB}^{1/2} \big(\rho_A^{-1/2}M_1^{\dagger}\otimes I_B\big)     \right) \left( \big(I_A\otimes N_1^{\dagger}\rho_B^{-1/2}\big) \rho_{AB}^{1/2}     \right)             \right]\\
& \leq \left[\tr\left( \rho_{AB} \big(\rho_A^{-1/2}M_1^{\dagger}M_1\rho_A^{-1/2}\otimes I_B\big)          \right)\right]^{1/2} \cdot \left[\tr\left( \rho_{AB} \big( I_A\otimes   \rho_B^{-1/2}N_1N_1^{\dagger}\rho_B^{-1/2}\big)          \right)\right]^{1/2}\\
& = \left[\tr\left( \rho_{A} \big( \rho_A^{-1/2}M_1^{\dagger}M_1\rho_A^{-1/2}  \big)        \right)\right]^{1/2} \cdot \left[\tr\left( \rho_{B}   \big(\rho_B^{-1/2}N_1N_1^{\dagger}\rho_B^{-1/2}  \big)        \right)\right]^{1/2}\\
& = \left[ \tr\big(   M_1^{\dagger}M_1  \big)  \right]^{1/2} \cdot \left[ \tr\big(   N_1N_1^{\dagger}   \big)  \right]^{1/2}\\
&= 1.
\end{align*}
On the other hand, observe that $\|\rho_A^{1/2}\|_2=\|\rho_B^{1/2}\|_2=1$ and $|\tr(\widetilde \rho_{AB} \rho_A^{1/2}\otimes \rho_B^{1/2})| =|\tr(\rho_{AB})|=1$. As a result, 
$$\lambda_1=1,$$
and we can take $M_1=\rho_A^{1/2}$ and $ N_1=\rho_B^{1/2}$.

Now for $R_A$ and $S_B$ satisfying~\eqref{eq:1} and~\eqref{eq:2} we have 
\begin{align*}
|\tr(\widetilde \rho_{AB} R_A\otimes S_B)|  & = \left|\sum_{i\geq 1} \lambda_i  \langle R^{\dagger}, M_i\rangle \langle S^{\dagger} , N_i  \rangle\right|\\
& = \left|\sum_{i\geq 2} \lambda_i  \langle R^{\dagger}, M_i\rangle \langle S^{\dagger} , N_i  \rangle\right|\\
&  \leq    \left(  \sum_{i\geq 2}     \lambda_i   | \langle R^{\dagger}, M_i\rangle  |^2    \right)^{1/2} \left(  \sum_{i\geq 2}     \lambda_i   | \langle S^{\dagger}, N_i\rangle  |^2    \right)^{1/2}  \\
&\leq \lambda_2,
\end{align*}
where in the last line we use $1=\|R\|_2^2 = \sum_{i\geq 2} | \langle R^{\dagger}, M_i\rangle  |^2$ and similarly 
$1= \sum_{i\geq 2}  | \langle S^{\dagger}, N_i\rangle  |^2 $.
These inequalities are tight for 
$R=M_2^{\dagger}$ and $S = N_2^{\dagger}$. We conclude that $\mu(\rho_{AB}) = \lambda_2$.

\end{proof}\\

Let us consider the special case where $A$ and $B$ are classical registers. If $\{\ket i: 1\leq i\leq d_A\}$ and $\{\ket k: 1\leq k\leq d_B\}$ are computational bases of $\mathcal{H}_A$ and $\mathcal{H}_B$ respectively, then $\rho_{AB}$ is diagonal in the basis $\{\ket i\ket k:  1\leq i\leq d_A, 1\leq k\leq d_B    \}$. Let us denote $p_{ik} = \bra i\bra k \rho_{AB}\ket i\ket k$, so we can think of a joint distribution $P_{AB}$ with marginals $P_A$ and $P_B$.
Then it is easy to see that 
\begin{align*}
\mu(P_{AB})   = \max\,\, & \mathbb E (f(i)g(k))\\
& \mathbb E(f(i)) = \mathbb E(g(k))=0,\\
& \mathbb E(f(i)^2) = \mathbb E(g(k)^2)=1,
\end{align*}
where the maximum is taken over all real functions $f$ and $g$ defined on $\{1, \dots, d_A\}$ and $\{1, \dots, d_B\}$ respectively, and $\mathbb E$ denotes the expectation value with respect to $P_{AB}$. 

Maximal correlation can be reformulated using Theorem~\ref{thm:schmidt}. Define 
$$\widetilde p_{ik} = p_i^{-1/2}p_k^{-1/2} p_{ik},$$
and let $\widetilde P_{AB}$ be a $d_A\times d_B$ matrix whose $ik$-th entry is $\widetilde p_{ik}$. It is easy to see that Schmidt coefficients of $\widetilde \rho_{AB}$ are in one-to-one correspondence with singular values of $\widetilde P_{AB}$. So in the classical case $\mu(P_{AB})$ is equal to the second singular value of $\widetilde P_{AB}$. For example if $P_{AB}$ denotes two perfectly correlated random variables, then $\widetilde P_{AB}$ is the identity matrix and $\mu(P_{AB})=1$. 
This latter formulation of maximal correlation in the classical case is found by Kang and Ulukus~\cite{KangUlukus} and Kumar~\cite{Kumar}.

\begin{theorem}\label{thm:ab} $\mu(\cdot)$ satisfies the following properties:
\begin{enumerate}
\item[{\rm(a)}] $\mu(\rho_{AB}\otimes \sigma_{A'B'}) = \max\{ \mu(\rho_{AB}), \mu(\sigma_{A'B'})  \}$.
\item[{\rm(b)}] Let $\Phi_{B}:\mathbf{L}(\mathcal{H}_B)\rightarrow \mathbf{L}(\mathcal{H}_{B'})$ be a completely positive trace-preserving super-operator. Let $\sigma_{AB'} = \mathcal{I}_A\otimes \Phi_B(\rho_{AB})$. Then
$\mu(  \sigma_{AB'}  )   \leq \mu(\rho_{AB})$.
\end{enumerate}
\end{theorem}

\begin{proof}
(a) Let $\lambda_1=1\geq \lambda_2\geq \cdots$ and $\zeta_1=1\geq \zeta_2\geq \cdots$ be the Schmidt coefficients of $\widetilde \rho_{AB}$ and $\widetilde \sigma_{A'B'}$ respectively. By Theorem~\ref{thm:schmidt}, $\mu(\rho_{AB}\otimes \sigma_{A'B'})$ is equal to the second Schmidt coefficient of $\widetilde \rho_{AB}\otimes \widetilde\sigma_{A'B'}$ which is equal to 
$$\max \{ \lambda_1\zeta_2,  \zeta_1 \lambda_2  \}  = \max\{ \lambda_2, \zeta_2  \} = \max \{\mu(\rho_{AB}), \mu(\sigma_{A'B'})\}.$$

\noindent (b) Any completely positive trace-preserving map is a composition of an isometry and a partial trace. Local isometries obviously do not change $\mu(\rho_{AB})$. Moreover $\mu(\sigma_{AB}) \leq \mu(\sigma_{AA'B})$ is easy to prove. Here we present a proof for the case where $\Phi_B$ is only 2-positive and not necessarily completely positive.

Let $X_A$ and $Y_{B'}$ be the optimizers for $\sigma_{AB'}$ satisfying \eqref{eq:con1} and \eqref{eq:con2}. Let $\Phi^{*}$ be the adjoint of $\Phi$, i.e., 
$\tr(\Phi(M)N) = \tr(M \Phi^*(N))$. Note that $\Phi^*$ is 2-positive since $\Phi$ is 2-positive, and $\Phi^*(I)= I$ because $\Phi$ is trace-preserving. Define $Z=\Phi^{*}(Y)$. Observe that $\sigma_A=\rho_A$ and $\sigma_{B'}= \Phi(\rho_B)$. Thus $\tr(\rho_A X) = \tr(\sigma_A X) =0$  and $\tr( \rho_AXX^{\dagger}  ) = \tr(\sigma_AXX^{\dagger})=1.$
Moreover, 
$$\tr(\rho_B Z) =\tr(\rho_B \Phi^{*}(Y)) = \tr(\Phi(\rho_B)Y) = \tr(\sigma_BY) = 0,$$ 
and 
\begin{align*}
\tr(\rho_{AB}X\otimes Z^{\dagger}) & = \tr(\rho_{AB}X\otimes \Phi^*(Y^{\dagger}))\\
& = \tr(\mathcal{I}_A\otimes \Phi_B(\rho_{AB})X\otimes Y^{\dagger})\\
& = \tr(\sigma_{AB'} X\otimes Y^{\dagger})\\
& = \mu(\sigma_{AB'}),
\end{align*}
where we use the fact that both $\Phi$ and $\Phi^*$ are hermitian-preserving. Therefore, we conclude that $\mu(\rho_{AB})\geq \mu(\sigma_{AB'})$ if $\tr(\rho_BZZ^{\dagger})\leq 1$. To prove the latter, observe that 
\begin{align*}
\begin{pmatrix}
YY^{\dagger} & Y\\
Y^{\dagger} & I
\end{pmatrix} = \begin{pmatrix}
Y\\
I
\end{pmatrix}\begin{pmatrix}
Y^{\dagger} & I
\end{pmatrix}
\end{align*}
is positive semidefinite. Since $\Phi^{*}$ is 2-positive,
\begin{align*}
\begin{pmatrix}
\Phi^*(YY^{\dagger}) & \Phi^*(Y)\\
\Phi^*(Y^{\dagger}) & \Phi^*(I)
\end{pmatrix} = \begin{pmatrix}
\Phi^*(YY^{\dagger}) & \Phi^*(Y)\\
\Phi^*(Y^{\dagger}) & I
\end{pmatrix},
\end{align*}
is positive semidefinite. 
This means that $\Phi^*(YY^{\dagger}) \geq \Phi^{*}(Y)\Phi^*(Y^{\dagger})$. Now using $\rho_B\geq 0$ we have
\begin{align*}
\tr(\rho_BZZ^{\dagger}) & =  \tr(\rho_B \Phi^*(Y)\Phi^*(Y^{\dagger}))\\
& \leq \tr(\rho_B \Phi^*(YY^{\dagger})) \\
& = \tr(\Phi(\rho_B) YY^{\dagger} )\\
& = \tr(\sigma_B YY^{\dagger})\\
& =1.
\end{align*}
We are done.

\end{proof}\\

The following corollary is the main result of this paper and is a direct consequence of the above theorem.

\begin{corollary} \label{cor:1} Suppose that $\rho_{AB}^{\otimes n}$, for some $n$, can be locally transformed to $\sigma_{EF}$ (under completely positive trace-preserving super-operators). Then
$$\mu(\rho_{AB}) \geq \mu(\sigma_{EF}).$$
\end{corollary}

\vspace{.15in}

The following example reveals the strength of this corollary. Let $\ket \psi_{AB} = \frac{1}{\sqrt{2}}(\ket{00} + \ket {11})$ be the Bell state on two qubits. Define 
$$\rho^{(p)}_{AB} = (1-p) \frac{I_{AB}}{4} + p\,\ket \psi\bra\psi_{AB},$$
where $0\leq p\leq 1$ and $I_{AB}/4$ is the maximally mixed state. Note that $\rho_A^{(p)} = I_A/2$ and $\rho_B^{(p)} = I_B/2$ for every $p$. Therefore,
\begin{align*}
\mu(\rho_{AB}^{(p)}) = \max \,\, & \tr(\rho_{AB}^{(p)}X\otimes Y^{\dagger})\\
& \tr(X) = \tr(Y)=0,\\
& \tr(XX^{\dagger}) = \tr(YY^{\dagger}) = 2.
\end{align*}
For $X$ and $Y$ satisfying the above equations we have 
\begin{align*}
|\tr(\rho_{AB}^{(p)}X\otimes Y^{\dagger})| & = \left|\frac{1-p}{4}\tr(X\otimes Y^{\dagger}) + p\bra \psi X\otimes Y^{\dagger} \ket \psi\right|\\
&=  \frac{p}{2}\, |\tr(X^TY^{\dagger})|\\
& \leq \frac{p}{2} \|X\|_2\cdot \|Y\|_2\\
& = p.
\end{align*}
This upper bound is achievable at $X=Y=\ket 0\bra 0 - \ket 1\bra 1$.  Therefore, 
$$\mu(\rho_{AB}^{(p)}) = p.$$
We conclude that entanglement cannot be distilled from $\rho_{AB}^{(p)}$ for $p< 1$ under local operations (but no communication) because for the maximally entangled state we have $\mu(\ket{\psi}\bra{\psi}_{AB})= \mu(\rho_{AB}^{(1)})  =1$. In fact even common randomness cannot be extracted from these states (under local operations) since for two perfectly correlated bits $U, V$, we have $\mu(P_{UV})=1$. Moreover, having infinitely may copies of $\rho_{AB}^{(p)}$ one cannot locally generate a single copy of $\rho_{AB}^{(q)}$ if $q>p$.

\section{Other Schmidt coefficients}

In this section we generalize the data processing inequality of the previous section for $\mu(\cdot)$ to other Schmidt coefficients of $\widetilde \rho_{AB}$. These new inequalities, however, do not hold in the $n$-letter case in the sense of  Corollary~\ref{cor:1}.

\begin{theorem} \label{thm:higher-coeff}
Let $1=\mu_1(\rho_{AB})\geq \mu(\rho_{AB})=\mu_2(\rho_{AB}) \geq \mu_3(\rho_{AB})\geq \cdots$ be Schmidt coefficients of $\widetilde\rho_{AB}$.
Let $\Phi_{B}:\mathbf{L}(\mathcal{H}_{B})\rightarrow \mathbf{L}(\mathcal{H}_{B'})$ be a completely positive trace-preserving super-operator and let $\sigma_{AB'} = \mathcal{I}_A\otimes \Phi_{B}(\rho_{AB})$. Then for every $i$ we have
$$\mu_i(\rho_{AB})\geq \mu_i(\sigma_{AB'}).$$
\end{theorem}

To prove this theorem it is more convenient to use the isomorphism between the Hilbert spaces $\mathcal{V}\otimes \mathcal{W}$ and $\mathbf{L}(\mathcal{V}, \mathcal{W})$, and the fact that Schmidt coefficients are mapped to singular values under this isomorphism. To be more precise, let us fix an orthonormal basis $\{\ket 0, \dots, \ket{d-1}\}$ for $\mathcal{H}_A$. Then for every $Z_{AB}\in \mathbf{L}(\mathcal{H}_A)\otimes \mathbf{L}(\mathcal{H}_B)$ there exists a super-operator 
$\Omega_Z: \mathbf{L}(\mathcal{H}_A)\rightarrow \mathbf{L}(\mathcal{H}_B)$ such that 
$$ Z_{AB} = \sum_{i,j=0}^{d-1}  \ket i\bra j_A\otimes \Omega_Z(\ket j\bra i)_B.$$
Using the fact that $\{\ket i\bra j:  i,j=0,\dots, d-1\}$ is an orthonormal basis for $\mathbf{L}(\mathcal{H}_A)$ it is easy to see that Schmidt coefficients of $Z_{AB}$ are equal to singular values of $\Omega_Z$.

By the above notation we may consider super-operators $\Omega_{\rho}$ and $\Omega_{\widetilde \rho}$. Observe that 
\begin{align*}
\sum_{i,j} \ket i\bra j \otimes \rho_{B}^{-1/2} \Omega_{\rho}(\rho_A^{-1/2} \ket j\bra i) & = \sum_{i,j, k} \ket i\bra j \otimes \rho_B^{-1/2} \Omega_\rho(\ket k\bra k \rho_A^{-1/2} \ket j\bra i)\\
& = \sum_{i,j, k} \ket i\bra k \rho_A^{-1/2} \ket j\bra j \otimes \rho_B^{-1/2} \Omega_\rho(\ket k\bra i)\\
& = \sum_{i, k} \ket i\bra k \rho_A^{-1/2} \otimes \rho_B^{-1/2} \Omega_\rho(\ket k\bra i)\\
&= (I_A\otimes \rho_B^{-1/2})   \left( \sum_{i,k} \ket i\bra k\otimes \Omega_{\rho}(\ket k\bra i)       \right)    (\rho_A^{-1/2}\otimes I_A)\\
&=(I_A\otimes \rho_B^{-1/2})   \rho_{AB}    (\rho_A^{-1/2}\otimes I_A)\\
& = \widetilde \rho_{AB}.
\end{align*}
Therefore, by definition we have 
\begin{align}\label{eq:omega-tilde}
\Omega_{\widetilde\rho}(X) = \rho_B^{-1/2}\Omega_\rho(\rho_A^{-1/2}X).
\end{align}

We are now ready to prove Theorem~\ref{thm:higher-coeff}.\\

\begin{proof}
From the definitions it is clear that $\Omega_\sigma= \Phi\circ \Omega_\rho$, and then from~\eqref{eq:omega-tilde} and $\sigma_A=\rho_A$ we have
\begin{align*}
\Omega_{\widetilde \sigma}(X) & = \sigma_{B'}^{-1/2} \Omega_{\sigma} (\sigma_A^{-1/2} X)\\
& = \sigma_{B'}^{-1/2} \Phi \left(\Omega_{\rho} (\rho_A^{-1/2} X)\right)\\
&=\sigma_{B'}^{-1/2} \Phi\left(  \rho_B^{1/2} \Omega_{\widetilde \rho}(X)    \right). 
\end{align*}
This means that if we define $\Psi:\mathbf{L}(\mathcal{H}_B)\rightarrow \mathbf{L}(\mathcal{H}_{B'})$ by $\Psi(Y) = \sigma_{B'}^{-1/2} \Phi(\rho_B^{1/2}Y)$ then 
$$\Omega_{\widetilde \sigma} = \Psi\circ \Omega_{\widetilde \rho}.$$
Thus given the correspondence between singular values and Schmidt coefficients we conclude that
$$\mu_i(\sigma_{AB'}) \leq \|\Psi\|\cdot \mu_i(\rho_{AB}),$$
where $\|\Psi\| = \|\Psi\|_{\infty}$ denotes the operator norm of $\Psi$, and we use (for example) Problem III.6.2 of~\cite{Bhatia-M}. Thus it suffices to show that $\|\Psi\|\leq 1$.

Fix an orthonormal basis $\{\ket 0, \ket 1, \dots, \ket{d'-1}\}$ for $\mathcal{H}_B$ and define 
$$\tau_{BB'} = \sum_{k,l=0}^{d-1} \ket k\bra l\otimes \Phi\left(   \rho_B^{1/2}\ket k\bra l\rho_B^{1/2}     \right).$$
It is easy to verify that $\tau_{BB'}$ is a density matrix with marginals $\tau_{B'}=\Phi(\rho_B)=\sigma_B$ and $\tau_B= \rho_B^{\ast}$ where by $\rho_B^{*}$ we mean the entry-wise complex conjugate of $\rho_B$ (with respect to the chosen basis). Moreover, we have $\Omega_{\tau}(X) =  \sigma_B^{-1/2}\Phi(\rho_B^{1/2} X^T)$ and then $\Omega_{\widetilde \tau} (X) = \Psi(X^T)$. As a result, $\|\Psi\| = \|\Omega_{\widetilde \tau}\|$ which we know is equal to the maximum Schmidt coefficient of $\widetilde \tau_{BB'}$ which is $1$.

\end{proof}

\section{Extreme values}

In this section we study the extreme values of maximal correlation.



\begin{theorem}\label{thm:extreme}
$0\leq \mu(\rho_{AB})\leq 1$ and the followings hold:
\begin{itemize} 
\item[\rm{(a)}] $\mu(\rho_{AB})=0$ if and only if $\rho_{AB}= \rho_A\otimes \rho_B$, i.e., $\rho_{AB}$ contains neither classical nor quantum correlation.
\item[\rm{(b)}] $\mu(\rho_{AB})=1$ if and only if there exist nontrivial local operators $X_A, Y_B$ such that $\rho_{AB}(X_A\otimes I_B) = \rho_{AB}(I_A\otimes Y_B)$. Furthermore, if $X_A, Y_B$ are hermitian, then there are local measurements $\{M_A, I_A-M_A\}$ and $\{N_B, I_B-N_B\}$ such that $\tr\left(\rho_{AB} M_A\otimes N_B\right)\neq 0,1$, and
$$\tr\left(\rho_{AB} (M_A\otimes (I_B- N_B))\right) = \tr\left( \rho_{AB} ((I_A- M_A) \otimes N_B)  \right) =0.$$

\end{itemize}
\end{theorem}

\begin{proof}  $\mu(\rho_{AB})\geq 0$ is clear from the definition and $\mu(\rho_{AB})=\lambda_2\leq \lambda_1=1$ follows from the proof of Theorem~\ref{thm:schmidt}. 

(a) $\mu(\rho_{AB})=0$ if and only if all Schmidt coefficients of $\widetilde \rho_{AB} $ except the first one ($\lambda_1=1$) are zero, which means that 
$$\widetilde \rho_{AB} = M_1\otimes N_1 = \rho_A^{1/2}\otimes \rho_B^{1/2},$$ 
or equivalently $\rho_{AB}=\rho_A\otimes \rho_B$. 

(b) Suppose that $\mu(\rho_{AB})=1$, so there exist $X_A$ and $Y_B$ satisfying~\eqref{eq:con1} and~\eqref{eq:con2}, and $\tr(\rho_{AB}X_A\otimes Y_B^{\dagger})=1$. We assume without loss of generality that $\rho_A$ and $\rho_B$ are invertible.
Define $Z_{AB} = X_A\otimes I_B - I_A\otimes Y_B$. Observe that
\begin{align*}
\tr(\rho_{AB} ZZ^{\dagger} ) 
& = \tr(\rho_{AB}XX^{\dagger}\otimes I_B) + \tr(\rho_{AB}  I_A\otimes YY^{\dagger})  - \tr(\rho_{AB} X\otimes Y^{\dagger} )- \tr(\rho_{AB} X^{\dagger}\otimes Y )\\
& = \tr(\rho_A XX^{\dagger}) + \tr(\rho_B YY^{\dagger})  - 2\\
&= 0. 
\end{align*}
Since both $\rho_{AB}$ and $ZZ^{\dagger}$ are positive semidefinite we conclude that $\rho_{AB}ZZ^{\dagger}=0$ and in fact
$\rho_{AB}Z=0$. Equivalently, we obtain 
\begin{align}\label{eq:x-to-y}
\rho_{AB}(X_A\otimes I_B) = \rho_{AB}(I_A\otimes Y_B).
\end{align} 

Conversely, if nontrivial operators $X_A, Y_B$ satisfying the above equation exists, then shifting them with identity we can assume that $\tr(\rho_A X_A)=\tr(\rho_{AB} (X_A\otimes I_B)) = \tr(\rho_{AB}(I_A\otimes Y_B)) = \tr(\rho_BY_B)=0$. Moreover, by scaling them we can assume that $\tr(\rho_A X_AX_A^\dagger) =1$ which implies
\begin{align*}
\tr(\rho_B Y_BY_B^\dagger) &= \tr\big( \rho_{AB}  (I_A\otimes Y_B)(I_A\otimes Y_B^\dagger)\big)\\
& = \tr\big( (I_A\otimes Y_B^\dagger)\rho_{AB}  (I_A\otimes Y_B)\big)\\
&= \tr\big( (I_A\otimes Y_B^\dagger)\rho_{AB}  (X_A\otimes I_B)\big)\\
&= \tr\big( (X^\dagger_A\otimes I_B)\rho_{AB}  (X_A\otimes I_B)\big)\\
& = \tr(\rho_A X_AX^\dagger_A)\\
& = 1,
\end{align*}
where here we use the fact that~\eqref{eq:x-to-y} implies  $(X_A^\dagger\otimes I_B)\rho_{AB} = (I_A\otimes Y^\dagger_B)\rho_{AB}$. Thus, $X_A, Y_B$ satisfy~\eqref{eq:con1} and~\eqref{eq:con2}, and we have 
$$\tr(\rho_{AB} X_A\otimes Y_B^\dagger) = \tr\big(\rho_{AB} (X_A\otimes I_B)(I_A\otimes Y_B^\dagger)\big) = \tr\big(\rho_{AB}  (I_A \otimes Y_BY_B^\dagger)\big)=\tr(\rho_BY_BY_B^\dagger) =1.$$

Suppose further that $X_A, Y_B$ are hermitian. Observe that
\begin{align*}
\rho_{AB}(X_A^2\otimes I_B) & = \rho_{AB}(X_A\otimes I_B) (X_A\otimes I_B) \\
& = \rho_{AB}(I_A\otimes Y_B) (X_A\otimes I_B) \\
& = \rho_{AB} (X_A\otimes I_B) (I_A\otimes Y_B)\\
& = \rho_{AB} (I_A\otimes Y_B)(I_A\otimes Y_B)\\
& = \rho_{AB} (I_A\otimes Y_B^2).
\end{align*}
More generally, for every polynomial $q(t)$ we have $\rho_{AB}(q(X_A)\otimes I_B) = \rho_{AB}(I_A\otimes q(Y_B))$. 
Using~\eqref{eq:con1}, $X_A$ is not a multiple of identity and has a non-trivial eigenspace. On the other hand orthogonal projections on eigenspaces of a hermitian operator can be written as polynomials in terms of that operator with \emph{real} coefficients. Therefore,  there exists a non-zero orthogonal projection $q(X_A)=M_A\neq I_A$ and a hermitian operator $q(Y_B)=N_B$ such that 
$$\rho_{AB}(M_A\otimes I_B) = \rho_{AB}(I_B\otimes N_B).$$ 
Replacing $N_A$ with $N_A^2$, we may assume that $N_B$ is positive semidefinite because 
$$\rho_{AB}(I_B\otimes N_A^2) = \rho_{AB}(M_A^2\otimes I_B) = \rho_{AB}(M_A\otimes I_B).$$
Note that
$$\tr(\rho_B N_B^n) = \tr(\rho_{AB} (I_A\otimes N_B^n)) = \tr(\rho_{AB}(M_A\otimes I_B))= \tr(\rho_AM_A).$$
Moreover, since $\rho_A$ is full-rank and $M_A$ is a non-trivial projection, $0< \tr(\rho_A M_A) < 1$. Now if $N_A$ has an eigenvalue greater than $1$, since $\rho_B$ is full-rank, $\tr(\rho_B N_B^n)$ would tend to infinity as $n$ goes to infinity. We conclude that all eigenvalues of $N_B$ are less than or equal to $1$ and $N_B\leq I_B$. 

Consider the local measurements $\{M_A, I_A-M_A\}$ and $\{N_B, I_B-N_B\}$ to be applied on $\rho_{AB}$. The probability of obtaining $M_A$ and $I_B-N_B$ is equal to 
\begin{align*}
\tr\left(\rho_{AB} (M_A\otimes (I_B-N_B))\right) & = \tr(\rho_{AB} (M_A\otimes I_B)) - \tr(\rho_{AB}(M_A\otimes N_B)) \\
& = \tr(\rho_{AB} (M_A\otimes I_B)) - \tr\left(\rho_{AB} (I_A\otimes N_B) (M_A\otimes I_B)    \right)\\
& = \tr(\rho_{AB} (M_A\otimes I_B)) - \tr\left(\rho_{AB} (M_A\otimes I_B) (M_A\otimes I_B)    \right)\\
& = \tr(\rho_{AB} (M_A\otimes I_B)) - \tr\left(\rho_{AB} (M_A\otimes I_B)    \right)\\
& = 0.
\end{align*}
Similarly we have $\tr\left( \rho_{AB}((I_A-M_A)\otimes N_B)   \right) =0$. We have $\tr(\rho_{AB}M_A\otimes N_B)\neq 0, 1$ because 
$$\tr(\rho_{AB}M_A\otimes N_B)  = \tr(\rho_{AB} M_A^2\otimes I_B) = \tr(\rho_AM_A)$$
is strictly between $0$ and $1$.

\end{proof}\\


\begin{remark}
When the marginal states $\rho_A, \rho_B$ are maximally mixed, the optimal operators $X_A, Y_B$ can be chosen to be symmetric. The point is that in this case $\widetilde \rho_{AB}$ is proportional to $\rho_{AB}$ that is hermitian and belongs to the tensor product of spaces of hermitian operators acting on subsystems $A$ and $B$ as \emph{real} Hilbert spaces.  Therefore, in this case the Schmidt decomposition of $\widetilde \rho_{AB}$ consists of Hermitian operators and following the proof of Theorem~\ref{thm:schmidt} it can be verified that the optimal operators $X_A, Y_B$ in the definition of quantum maximal correlation can be chosen to be hermitian.   
\end{remark}

\medskip

\begin{example}\label{example:hermitian} 
In this example we show that optimal $X_A, Y_B$ are not hermitian in general. Let 
$$\rho_{AB} = \frac{1}{\sqrt 2} \big(\ket{00}\bra{00} + \ket{++}\bra{++}\big),$$
where $\ket + = \frac{1}{\sqrt 2} (\ket 0+\ket 1)$. Let $X=Y=\ket 0\bra 0 -2 \ket 0\bra 1 - \ket 1\bra 1$. Then, it is not hard to verify that $\tr(\rho_A X_A) = \tr(\rho_B Y_A) =0$, $\tr(\rho_A X_A^\dagger X_A)= \tr(\rho_B Y_B^\dagger Y_B)=1$ and $\tr(\rho_{AB} X_A\otimes  Y^\dagger_B)=1$. Therefore, $\mu(A, B)=1$. Note that $X_A, Y_B$ are not hermitian and no other such hermitian operators exist that satisfy the above equations. To prove the latter claim, if such hermitian operators exists, then there are local measurements $\{M_A, I_A-M_A\}$ and $\{N_B, I_B-N_B\}$ that satisfy part (b) of Theorem~\ref{thm:extreme}. Nevertheless, it is not hard to verify that such local measurements do not exist \footnote{This example is shared with us by Saleh Rahimi-Keshari.}.

\end{example}

\medskip

We say that $\rho_{AB}$ has a common data (in the asymptotic sense) if for every $\epsilon>0$ there exists $n$ and local measurements $\{M_0, M_1=I-M_0\}$ and $\{N_0, N_1=I-N_0\}$ with outcomes $U, V$ such that 
$$\Pr(U\neq V)=\tr(\rho_{AB}^{\otimes n} M_0\otimes N_1) + \tr(\rho_{AB}^{\otimes n} M_1\otimes N_0)\leq \epsilon,$$
and $\Pr(U=0, V=0), \Pr(U=1, V=1)\geq c$ where $c>0$ is a constant independent of $n$ and $\epsilon$.

In the classical case a joint distribution $P_{UV}$ has a common data if it is \emph{decomposable}. Decomposability means that $\mathcal U$ and $\mathcal V$, the ranges (set of alphabets) of $U, V$, can be decomposed as \emph{disjoint} unions $\mathcal U = \mathcal U_0 \cup \mathcal U_1$ and $\mathcal V = \mathcal V_0 \cup \mathcal V_1$ such that 
$\Pr [ \mathcal U_i  \times \mathcal V_j]$
is equal to zero if $i\neq j$ and is positive otherwise. It is shown in~\cite{Witsenhausen} that decomposability is equivalent to having a common data even in the asymptotic sense, and that both of these are equivalent to $\mu(P_{UV})=1$. Here we attempt to generalize this to the quantum case.

\begin{lemma}\label{lem:contin} Let $U, V$ be two binary random variables such that $p_{01}, p_{10}\leq \epsilon$. Then 
$$\mu(P_{UV}) \geq 1 - \frac{\epsilon}{p_{00}p_{11}} - \frac{2\epsilon^2}{p_{00}p_{11}}.
$$
\end{lemma} 

\begin{proof}
$\mu(P_{UV})$ is equal to the second singular value of
\begin{align*}
\widetilde P_{UV}=
\begin{pmatrix}
\frac{p_{00}}{\sqrt{(p_{00}+p_{01})(p_{00}+p_{10})}} & \frac{p_{01}}{\sqrt{(p_{00}+p_{01})(p_{01}+p_{11})}}\\
\frac{p_{10}}{\sqrt{(p_{10}+p_{11})(p_{00}+p_{10})}} & \frac{p_{11}}{\sqrt{(p_{10}+p_{11})(p_{01}+p_{11})}}
\end{pmatrix}.
\end{align*}
We know that the first singular value of $P_{UV}$ is $1$, so the second singular value is equal to 
\begin{align*}
\mu(P_{UV}) & = |\det\, \widetilde P_{UV}|\\
& = \frac{|p_{00}p_{11} - p_{01}p_{10}|}{ \sqrt{(p_{00}+p_{01})(p_{00}+p_{10})(p_{10}+p_{11})(p_{01}+p_{11})}  }\\
& \geq \frac{p_{00}p_{11}}{(p_{00}+\epsilon)(p_{11}+\epsilon)}   -\frac{p_{01}p_{10}}{p_{00}p_{11}}\\
& \geq 1- \frac{\epsilon(p_{00}+p_{11}) + \epsilon^2}{p_{00}p_{11}} - \frac{\epsilon^2}{p_{00}p_{11}}\\
& \geq 1 - \frac{\epsilon}{p_{00}p_{11}} - \frac{2\epsilon^2}{p_{00}p_{11}}.
\end{align*}

\end{proof}

\begin{theorem}\label{thm:common-data} $\rho_{AB}$ has a common data (in the asymptotic sense) only if $\mu(\rho_{AB})=1$.
\end{theorem}

\begin{proof} If $\rho_{AB}$ has a common data in the asymptotic sense, then by definition for every $\epsilon>0$ and sufficiently large $n$, $\rho_{AB}^{\otimes n}$ under local measurements can be transformed to random variables $U$ and $V$ such that $\Pr(U\neq V)\leq \epsilon$ and $\Pr(U=0, V=0), \Pr(U=1, V=1)\geq c$ for some constant $c>0$ that is independent of $n$ and $\epsilon$.  Thus, using Lemma~\ref{lem:contin} and Corollary~\ref{cor:1} we have
$$\mu(\rho_{AB}) \geq 1 - \frac{\epsilon}{c^2} - \frac{2\epsilon^2}{c^2}.$$
The claim follows since this inequality holds for all $\epsilon>0$.
\end{proof}\\

In the example at the end of Section~\ref{sec:new} we see that maximal correlation can indeed take any value between $0$ and $1$. On pure states however it takes only the extreme values.

\begin{proposition}\label{prop:pure}
\begin{itemize}

\item[\rm{(i)}] Suppose $\rho_{AB}$ is pure. Then $\mu(\rho_{AB})=0$ if $\rho_{AB}$ is separable and $\mu(\rho_{AB})=1$ if $\rho_{AB}$ is entangled.

\item[\rm{(ii)}] If $\|  \rho_{AB} - \tau_{AB} \|_{\text{\emph{tr}}} \leq \epsilon$ where $\tau_{AB}$ is a maximally entangled state and $\epsilon \leq 1/10$, then 
$\mu(\rho_{AB}) \geq 1- 9 \epsilon$.

\end{itemize}
\end{proposition}

\begin{proof} (i) If $\rho_{AB}$ is separable, $\mu(\rho_{AB})=0$ follows from part (a) of Theorem~\ref{thm:extreme}. Thus suppose $\rho_{AB}=\ket \psi\bra \psi_{AB}$ where $\ket \psi_{AB}$ is entangled with Schmidt decomposition 
$$\ket\psi_{AB} = \sum_i \alpha_i \ket{v_i}_A\otimes \ket{w_i}_B.$$
Since $\ket \psi_{AB}$ is entangled at least two of the Schmidt coefficients (say) $\alpha_1$ and $\alpha_2$ are non-zero.
Define
$$X_A = c\alpha_2^2\ket {v_1}\bra {v_1} - c\alpha_1^2\ket {v_2}\bra{v_2},$$
and 
$$Y_A = c\alpha_2^2\ket {w_1}\bra {w_1} - c\alpha_1^2\ket {w_2}\bra{w_2},$$
where $c^{-1}=\alpha_1\alpha_2(\alpha_1^2+\alpha_2^2)^{1/2}$. Then $X_A$ and $Y_B$ satisfy~\eqref{eq:con1}, \eqref{eq:con2} and $\tr(\rho_{AB} X_A\otimes Y_B)=1$. As a result  $\mu(\rho_{AB})=1.$

(ii) Since $\tau_{AB}$ is a maximally entangled state, there are local measurements $\{M_A^0, M_A^1\}$ and $\{N_B^0, N_B^1\}$ such that $\tr(\tau_{AB}M_A^0\otimes N_B^1) = \tr(\tau_{AB}M_A^1\otimes N_B^0)=0$ and 
$$\tr(\tau_{AB}M_A^0\otimes N_B^0) = \frac{1}{d} \left\lfloor \frac{d}{2}\right\rfloor, \quad\quad   \tr(\tau_{AB}M_A^1\otimes N_B^1) = \frac{1}{d} \left\lceil \frac{d}{2}\right\rceil,$$
where $d$ is the minimum of the dimensions of registers $A$ and $B$. Let 
$$p_{uv}  =\tr(\rho_{AB}M_A^u\otimes N_B^v).$$
Then, using $\|\rho_{AB}- \tau_{AB}\|_{\tr}\leq \epsilon$ we find that 
\begin{align*}
\left| p_{00}-\frac{1}{d} \left\lfloor \frac{d}{2}\right\rfloor\right| + \left| p_{11} - \frac{1}{d} \left\lceil \frac{d}{2}\right\rceil  \right| + p_{01} + p_{10} \leq \epsilon,
\end{align*}
which using $\epsilon\leq 1/10$ implies 
\begin{align}\label{eq:rt1}
p_{00}p_{11}\geq \left(\frac{1}{d} \left\lfloor \frac{d}{2}\right\rfloor -\epsilon \right)\left( \frac{1}{d} \left\lceil \frac{d}{2}\right\rceil - \epsilon\right) \geq  \left(\frac{1}{3}-\epsilon\right)\left(\frac{2}{3}-\epsilon\right)\geq \frac{7\times 17}{30^2},
\end{align} 
and 
$p_{01} + p_{10}< \epsilon$.
By Corollary~\ref{cor:1} we have $\mu(\rho_{AB}) \geq \mu(P_{UV})$. So it suffices to show that $\mu(P_{UV}) \geq 1-9\epsilon$ which is a simple consequence of Lemma~\ref{lem:contin}. 

\end{proof}

Although part (ii) of this proposition states the continuity of $\mu(\cdot)$ at maximally entangled states, by part (i) it takes values $0$ or $1$ on pure states and is not continuous at separable states. We note that such a seemingly undesirable property is unavoidable. The main point is that for every entangled pure state
$\ket\psi_{AB}$, its $n$-fold tensor product $\ket{\psi}_{AB}^{\otimes n}$ is close to a maximally entangled state.  
Then, considering the continuity of $\mu(\cdot)$ at maximally entangled states, we conclude that $\mu(\ket \psi\bra\psi_{AB})= 1$.  

\section{Concluding remarks}

In this paper we generalized a measure of bipartite correlation called maximal correlation to quantum states.
We showed that this measure satisfies  $\mu(\rho_{AB}^{\otimes n})= \mu(\rho_{AB})$, and proved a data processing type inequality for it. We note that, being a Schmidt coefficient, the quantum maximal correlation can be computed efficiently.

$\mu(\cdot)$ is a measure of total classical and quantum correlations and takes its maximum value on two perfectly correlated bits. This implies that $\mu(\cdot)$ is not well-behaved under (even one bit of) classical communication. It is tempting to look for a measure of \emph{quantum} correlation with the above properties that vanishes on classically correlated states. In that case we could use such a measure to study the problem of entanglement distillation under LOCC maps. Even proving the nonexistence of such a measure would be interesting. \\

\noindent\textbf{Acknowledgements.} The author is thankful to Amin Gohari for introducing~\cite{KangUlukus}, to  Yury Polyanskiy for sending a copy of~\cite{Polyanskiy}, and to Saleh Rahimi-Keshari for sharing Example~\ref{example:hermitian}. The author is also grateful to the unknown referee who pointed an error in the statement of Theorem~\ref{thm:common-data}.
This research was in part supported by National Elites Foundation and by a grant from IPM (No.\ 91810409).


\end{document}